\documentclass[letterpaper, 10 pt, conference]{ieeeconf}

\IEEEoverridecommandlockouts                              %
\usepackage{cite}

\usepackage[utf8]{inputenc}
\usepackage{textcomp}
\usepackage{graphicx}
\usepackage{mathrsfs} 
\usepackage{amsmath,amssymb,amsfonts,graphicx,float}

\usepackage{mathtools}
\usepackage[version=4]{mhchem}
\usepackage{siunitx}
\usepackage{longtable,tabularx}
\usepackage{verbatim}

\usepackage{hyperref}
\usepackage{algorithm}
\usepackage{algpseudocode}

\usepackage{xcolor}

\usepackage{subcaption}

\usepackage{eso-pic}

\newcommand{\copyrightfooter}{%
  \AddToShipoutPictureFG*{%
    \AtPageLowerLeft{%
      \put(\LenToUnit{0.7in},\LenToUnit{0.25in}){%
        \parbox[b]{\dimexpr\paperwidth-1.4in\relax}{%
          \fontsize{7}{8}\selectfont
          \textcopyright~2026 IEEE.
          Personal use of this material is
permitted. Permission from IEEE must be obtained for all
other uses, in any current or future media, including
reprinting/republishing this material for advertising or
promotional purposes, creating new collective works,
for resale or redistribution to servers or lists,
or reuse of any copyrighted component of this work
in other works.
        }%
      }%
    }%
  }%
}

\newtheorem{assumption}{Assumption}
\newtheorem{theorem}{Theorem}
\newtheorem{lemma}{Lemma}

\newtheorem{remark}{Remark}

\title{\LARGE \bf Integrating Prescribed Performance and RISE Control for Asymptotic Tracking of Uncertain High-Order Nonlinear Systems}

\author{Athanasios K. Gkesoulis\thanks{The first two authors contributed equally to this work.}, Christos K. Verginis, George C. Karras and Charalampos P. Bechlioulis
\thanks{This work was supported by the Swedish Research Council (VR, grant 2024-04905)
and the project “Applied Research for Autonomous Robotic Systems” (MIS 5200632), which is implemented within the framework of the National Recovery and Resilience Plan “Greece 2.0” (Measure: 16618-Basic and Applied Research) and funded by the European Union--NextGenerationEU.}
\thanks{Athanasios K. Gkesoulis is with the Robotics Institute, Athena Research Center, Marousi, Greece (e-mail: a.gkesoulis@athenarc.gr).}
\thanks{Christos K. Verginis is with the Division of Signals and Systems, Department of Electrical Engineering, Uppsala University, Uppsala, Sweden (e-mail: christos.verginis@angstrom.uu.se).}
\thanks{George C. Karras is with the Robotics Institute, Athena Research Center,  Marousi, Greece, and also with the Department of Informatics and Telecommunications, School of Sciences, University of Thessaly, Lamia, Greece (e-mail: karrasg@athenarc.gr).}
\thanks{Charalampos P. Bechlioulis is with the Robotics Institute, Athena Research Center, 151 25 Marousi, Greece, and also with the Division of Systems and Automatic Control, Department of Electrical and Computer Engineering, University of Patras, Patras, Greece (e-mail: chmpechl@athenarc.gr).}}%

\begin{document}

\maketitle
\thispagestyle{empty}
\pagestyle{empty}

\copyrightfooter

\begin{abstract}
This paper addresses asymptotic tracking with prescribed transient and steady-state performance for a class of uncertain high-order nonlinear systems with uncertain dynamics. The proposed controller combines a barrier-type transformation of a filtered tracking error with a robust integral of the sign of the error (RISE) mechanism, yielding a continuous and low-complexity design. The proposed method does not employ neural, fuzzy, or other global approximation structures. It is shown that all  closed-loop signals remain bounded, the prescribed performance funnel is preserved for all initial conditions in the funnel, and that the tracking error converges asymptotically to zero under appropriate gain selection. The analysis provides a continuous-feedback synthesis that integrates prescribed performance control (PPC) with the RISE methodology. Simulation results on a ship-steering example illustrate the effectiveness of the proposed approach.

\end{abstract}

\section{Introduction} \label{sec:intro}

Nonlinear systems with uncertain dynamics, such as unmodeled nonlinearities, parametric uncertainty, or bounded disturbances, are unavoidable in practice. For such systems, prescribed performance control \cite{4639441} and funnel control \cite{ilchmann2002tracking} (FC) have emerged as central methodologies for enforcing transient and steady-state specifications, such as prescribed overshoot, convergence rate, and steady-state error bounds, through designer-selected performance envelopes and suitable barrier error transformations. Their appeal lies in the fact that they directly impose quantitative performance requirements under limited model information and often with relatively low implementation complexity \cite{bu2023prescribed, berger2025funnel}. Nevertheless, these methodologies typically guarantee confinement of the tracking error within a prescribed funnel, and convergence to a desired residual set, rather than asymptotic convergence to zero. In contrast, approaches based on the robust integral of the sign of the error \cite{1310482,xian2004continuous} provide robustness against uncertain nonlinear dynamics and constitute a standard route to asymptotic tracking with continuous control laws, but do not provide explicit transient performance guarantees of the PPC or FC type. The development of low-complexity continuous controllers that combine prescribed transient performance with asymptotic convergence to zero remains an important open problem.

Several recent works have addressed the simultaneous achievement of prescribed performance and asymptotic tracking for uncertain nonlinear systems. A discontinuous switching hybrid controller was considered in \cite{kanakis2019guaranteeing} for high-order nonlinear systems. In \cite{verginis2020asymptotic} a discontinuous control design was developed that achieves asymptotic tracking with prescribed transient response for a class of second-order unknown systems. Funnel control along with a barrier Lyapunov function  was utilized in \cite{liu2022asymptotic} in a recursive adaptive design. Asymptotic tracking via funnel control was also established in \cite{lee2019asymptotic} for a class of relative-degree one systems, by employing an alternative funnel-control formulation that allows the funnel boundary to converge to zero. Asymptotic prescribed performance schemes based on RISE have also been reported, typically together with adaptive or approximation mechanisms. Neural networks are employed along with RISE in \cite{fan2017asymptotic} for the control of high-order uncertain systems, in \cite{wang2017rise} for the control of servo mechanisms, in \cite{deng2022neural} for the control of hydraulic manipulators, and in \cite{cao2023neuroadaptive} for strict-feedback systems with bounded control gain matrix eigenvalues. Work \cite{cao2026achieving} combines RISE and PPC to accommodate generalized, non-monotonic performance specifications, utilizing  knowledge of the system dynamics in the control design. The barrier integral control technique was introduced in \cite{verginis2024asymptotic} for nonlinear systems with bounded control gain matrices, and further studied in \cite{verginis2024barrier} for reference trajectories with derivatives that converge to zero and no disturbances at steady state. The aforementioned results clearly highlight the significance of designing control laws for asymptotic tracking under prescribed performance requirements. However, despite their progress, a low-complexity continuous control design for high-order uncertain nonlinear systems that avoids neural or fuzzy approximation mechanisms and guarantees prescribed performance and asymptotic convergence of the tracking error remains  unaddressed.

It is worth noting that the aforementioned approaches come with important limitations. The designs in \cite{kanakis2019guaranteeing,verginis2020asymptotic} are discontinuous and may therefore induce chattering. Zero-converging funnels \cite{lee2019asymptotic} can exhibit similar numerical issues that stem from the product of ``large" and ``small" quantities (the funnel reciprocal and the error
signal). 
The RISE-based prescribed performance designs in \cite{fan2017asymptotic,wang2017rise,deng2022neural,cao2023neuroadaptive}, rely on neural or related approximation mechanisms and are therefore not of low implementation complexity. Moreover, their results are established over gain-dependent sets of initial conditions and therefore are semi-global, and address restrictive system classes. The approach of \cite{cao2026achieving}, although approximation-free, requires information about the dynamics. Finally, \cite{verginis2024asymptotic} addresses a stabilization problem for a limited class of systems, whereas \cite{verginis2024barrier} considers tracking under a restricted class of reference trajectories, requiring additional asymptotic conditions on the reference and disturbance terms.

Motivated by the above discussion, this paper considers asymptotic tracking with prescribed transient and steady-state performance for a class of uncertain high-order nonlinear systems. {We employ a barrier-type prescribed-performance transformation of a filtered tracking error within a RISE formulation to develop a continuous, low-complexity control law that enforces funnel-type transient behavior and asymptotic convergence of the tracking error.} In contrast to the discontinuous designs of \cite{kanakis2019guaranteeing,verginis2020asymptotic}, the resulting plant input is continuous, thus avoiding the possibility of induced chattering. Compared with the RISE-based prescribed-performance schemes of \cite{fan2017asymptotic,wang2017rise,deng2022neural,cao2023neuroadaptive}, the proposed design adopts less restrictive assumptions on the input gain, does not employ approximation structures, and establishes performance for all funnel-compatible initial conditions under the stated gain conditions rather than over gain-dependent initial sets. Unlike \cite{cao2026achieving}, our approach does not require knowledge of the system dynamics. Unlike \cite{lee2019asymptotic}, where asymptotic tracking is achieved for relative-degree-one systems by allowing the funnel boundary to converge to zero, the proposed design considers high-order nonlinear systems and does not require a vanishing funnel, thereby avoiding possible large inputs or numerical issues. Moreover, \cite{verginis2024asymptotic} addresses stabilization rather than trajectory tracking and assumes bounded input gains, unlike the present work. Finally, \cite{verginis2024barrier} assumes reference trajectories with derivatives converging to zero and no disturbances at steady state, while neither condition is imposed here.

The main contributions of this work can be summarized as follows. 
\begin{itemize}
    \item We propose a new control design that achieves asymptotic tracking with prescribed transient and steady-state performance for a class of uncertain high-order nonlinear systems.
    \item The proposed design provides a direct integration of prescribed performance barrier-type transformations with a RISE-type integral mechanism, thereby explicitly bridging the PPC methodology with asymptotic tracking control.
    \item  We establish boundedness of all closed-loop signals,  invariance of the prescribed performance funnel for all controller gains and all initial conditions satisfying the funnel constraints, and asymptotic convergence of the tracking error to zero under appropriate gain selection. In this sense, the proposed design extends the semi-global asymptotic results of existing works. %
 
    \item  The proposed controller is of low-complexity in the sense that it does not require observers or global approximators such as neural networks or fuzzy systems.
\end{itemize}

The rest of the paper is organized as follows. Section~\ref{sec:prelim} introduces the notation and preliminaries. 
Section~\ref{sec:pf} formulates the problem. Section~\ref{sec:main} presents the control design and analysis. Section~\ref{sec:sim} provides simulation results.  Section~\ref{sec:conclusion} concludes the paper.

\section{Preliminaries}\label{sec:prelim}

{The Euclidean norm is denoted by $\|\cdot\|$, $L_2$ and $L_\infty$ are the usual Lebesgue spaces, and ``a.e.'' means almost everywhere. For a measurable, locally essentially bounded vector field $h:\mathcal{D} \times [t_0,\infty) \to \mathbb{R}^n$, $\mathcal{D}\subset \mathbb{R}^n$, its Filippov regularization is \cite{filippov1988}}
\begin{align}
{
\mathsf K[h](z,t):=
\bigcap_{\delta>0}\,
\bigcap_{\mu(N)=0}
\overline{\operatorname{co}}\,
h\big(\mathcal B(z,\delta)\setminus N,t\big),}
\label{eq:Filippov-regularization}
\end{align}
{where $\mathcal B(z,\delta)$ is the open Euclidean ball, $\mu$ is the Lebesgue measure, and $\overline{\operatorname{co}}$ is the closed convex hull. An absolutely continuous $z:I\to\mathbb R^n$ is a Filippov solution if $\dot z(t)\in\mathsf K[h](z(t),t)$ a.e. If $V\in C^1$, the ordinary chain rule holds a.e. $\dot{\widetilde V}$ denotes its set of possible values along the regularized dynamics. With a slight abuse of notation, $\mathsf K[\dot r]$ denotes the regularization of the discontinuous right-hand side defining $\dot r$. Finally, $D^+r(t):=\limsup_{a\downarrow0}(r(t+a)-r(t))/a$ is the upper right Dini derivative.}

\section{Problem Formulation} \label{sec:pf}

{Consider a system of order $k\in\mathbb N$ with dynamics}
\begin{subequations} \label{eq:system}
\begin{align} 
	\dot{x}_i &= x_{i+1}, \ i\in\{1,\dots,k-1\} \\
	\dot{x}_k &= f(x,t) + g(x,t)u \\
        y &= x_1 
\end{align}
\end{subequations}
where  $x \coloneqq [x_1,\dots,x_k]^\top \in \mathbb{R}^k$ is the state,  $u \in\mathbb{R}$ is the control input, $y \in\mathbb{R}$ is the output, and $f:\mathbb{R}^{k}\times\mathbb{R}_{\geq 0}\to \mathbb{R}$, $g:\mathbb{R}^{k}\times\mathbb{R}_{\geq 0}\to\mathbb{R}$ are \textit{unknown} functions.
{Assume that $f,g\in C^2(\mathbb R^k\times\mathbb R_{\geq0})$ and that, for every compact $X\subset\mathbb R^k$, their partial derivatives in $(x,t)$ of total order at most two, including order zero, are bounded on $X\times\mathbb R_{\geq0}$.}
{The unknown drift $f$ represents modeling uncertainties, unknown parameters, and bounded exogenous effects, whereas $g$ is the uncertain input gain.} Further, we make the following assumption on $g(x,t)$, which describes a sufficient controllability condition for \eqref{eq:system} and is adopted in a wide variety of works (e.g., \cite{verginis2020asymptotic,verginis2024barrier,bechlioulis2016decentralized,xian2004continuous}).

\begin{assumption} \label{ass:g min}
{There exists $g_{\min}>0$ such that $g_{\min}\le g(x,t)$ for all $(x,t)\in \mathbb{R}^k\times \mathbb{R}_{\ge 0}$.}
\end{assumption}

{Let $y_d\in C^{k+2}(\mathbb R_{\geq0})$ denote the desired reference and define $y_d^{(i)}:=\mathrm d^iy_d/\mathrm dt^i$. Assume that $y_d^{(i)}$ is bounded for $i\in\{0,\ldots,k+2\}$. Define $e_i:=x_i-y_d^{(i-1)}$, $i\in\{1,\ldots,k\}$. Then $\dot e_i=e_{i+1}$ for $i\in\{1,\ldots,k-1\}$ and $\dot e_k=f(x,t)+g(x,t)u-y_d^{(k)}$. In particular, $e_1=y-y_d$.} The control objective is:
\begin{itemize}
\item  achieve asymptotic tracking, i.e., $\lim_{t\to\infty} e_1(t)=0$, using a {continuous plant input}, and
\item enforce prescribed performance bounds, i.e., guarantee that $|e_1(t)|<\phi(t)$ for all $t\ge 0$, where $\phi:\mathbb{R}_{\ge 0}\to\mathbb{R}_{>0}$ is a designer-chosen performance function that is positive, continuously differentiable, strictly decreasing, with a positive limit $\lim_{t \to \infty}\phi(t)=\phi_{\infty}>0$ and $\phi(0) > |e_1(0)|$.
\end{itemize}

\section{Main results}\label{sec:main}

\subsection{Control design}

{To solve the aforementioned problem, we propose a novel
integration of PPC and RISE. Our design employs (i) reciprocal barrier functions that retain the error $e_1(t)$ within the funnel dictated by $\phi(t)$, and (ii) the integral of the sign of $e_1(t)$ that enforces asymptotic tracking by implicitly compensating for the upper bounds of the drift dynamics within the funnel. Unlike \cite{verginis2020asymptotic}, the proposed algorithm
features \textit{continuous} feedback. Moreover, compared with
\cite{verginis2024asymptotic,verginis2024barrier}, it accommodates
a broader class of dynamics and reference trajectories.}

{Let $\bar e=[e_1,\ldots,e_k]^\top=x-\bar y_d$, where $\bar y_d=[y_d,\dot y_d,\ldots,y_d^{(k-1)}]^\top$.}
We define the auxiliary filtered errors $\bar{\eta} \coloneqq [\eta_1,\dots,\eta_k]^\top$ as
\begin{align} \label{eq:eta_i}
\eta_i &\coloneqq \left( \frac{\textup{d}}{\textup{d}t} + \lambda \right)^{i-1}e_1 \ = 
 \  \sum_{\ell=0}^{i-1} \left( \begin{matrix}
i- 1\\
\ell
\end{matrix} \right)
\lambda^\ell e_{i-\ell} 
\end{align}
{for all $i\in\{1,\dots,k\}$ and a constant $\lambda>0$, where the power denotes repeated application of the linear operator $\mathrm d/\mathrm dt+\lambda$. These filtered errors satisfy}
\begin{align} \label{eq:eta_i+1}
\eta_{i+1} = \dot{\eta}_i + \lambda \eta_i
\end{align}
for all $i\in\{1,\dots,k-1\}$.
In the sequel, we use the notation $\bar{e}_i \coloneqq [e_1,\dots,e_i]^\top$, $\bar{\eta}_i \coloneqq [\eta_1,\dots,\eta_i]^\top \in \mathbb{R}^i$, for $i\in\{1,\dots,k\}$, with $\bar{e}_1 = e_1$, $\bar{\eta}_1 = \eta_1$, $\bar{e}_k = \bar{e}$, $\bar{\eta}_k = \bar{\eta}$. By recursively applying \eqref{eq:eta_i}, we conclude that 
{$e_i = \sum_{j=1}^i \beta_{ij} \eta_j$, where $\beta_{ij}=\binom{i-1}{j-1}(-\lambda)^{i-j}$, for all $i \in \{1,\dots,k\}$.}
 The aforementioned expression
 defines the mappings  
 \begin{align} \label{eq:Beta mapping}
 \bar{e}_i = B_i \bar{\eta}_i
 \end{align}
 {where $B_i \in \mathbb{R}^{i \times i}$ is lower triangular with $[B_i]_{rj}=\beta_{rj}$ for $j\leq r$ and zero otherwise.}

We now describe the proposed control design. 
First, in order to satisfy the prescribed performance bounds $|e_1(t)| < \phi(t)$, we aim to bound the filtered error $\eta_k(t)$ in a corresponding prescribed funnel. 
{Choose a twice continuously differentiable, decreasing function $\rho:\mathbb{R}_{\geq0}\to\mathbb{R}_{>0}$ with bounded first and second derivatives such that, for constants $\rho_{\min},\rho_{\max}>0$, $\rho_{\min}\leq\rho(t)\leq\rho_{\max}$ and $\rho(0)>|\eta_k(0)|$.}
We then define the normalized and transformed errors 
\begin{align}\label{eq:ksi and epsilon}
    \xi \coloneqq \frac{\eta_k}{\rho(t)},~\varepsilon \coloneqq \mathbb{T}(\xi) \coloneqq \frac{1}{2}\ln\left(\frac{1+ \xi }{1-\xi}\right).
\end{align}
The {barrier-type} map $\mathbb{T}:(-1,1)\to\mathbb{R}$ is an increasing function with the property $\lim_{\ast \to \pm 1}\mathbb{T}(\ast) = \pm \infty$. Such a property is exploited in the control design, which, except for asymptotic tracking, aims to guarantee $\xi(t) \in (-1,1)$ by retaining the boundedness of $\varepsilon(t)$. Note that $\mathbb{T}^{-1}(\ast) = \tanh(\ast)$
{and that bounded $\varepsilon$ implies $|\xi|=|\tanh(\varepsilon)|<1$, equivalently $|\eta_k(t)|<\rho(t)$.}

We now design the control law as 
\begin{align} \label{eq:control input}
    u(t) = -(k_1 &+ 1)(\varepsilon(t) - \varepsilon(0))\nonumber\\ &- \int_0^t \big((k_1 + 1)\alpha  \varepsilon(\tau) + k_2\textup{sgn}(\varepsilon(\tau)) \big)\textup{d}\tau, 
\end{align}
where $k_1$, $\alpha$, $k_2$ are positive control gains and $\textup{sgn}(\cdot)$ denotes  the signum function.

The prescribed bound $|\eta_k(t)|<\rho(t)$ is translated to a bound for $e_1(t)$ as follows:
\begin{lemma}\cite{verginis2024barrier} \label{lemma:1}
Assume that $|\eta_k(t)| < \bar{\rho}$ for a constant $\bar{\rho} > 0$ and all $t \geq 0$. Then, there exist positive constants $\bar{e}_{i,1}$, $\bar{e}_{i,2}$, and $\lambda_0 \in (0,\lambda)$ such that, for every $i\in\{1,\dots,k-1\}$ and all $t\geq 0$:
\begin{align}
|e_i(t)| \leq \bar{e}_{i,1} \exp(-\lambda_0 t) + \bar{e}_{i,2}.
\end{align}
\end{lemma}

The constants $\bar e_{i,1}$, $\bar e_{i,2}$, and $\lambda_0$ depend on the initial errors, $\lambda$, and the bound $\bar\rho$ induced by $\rho(t)$. Hence $\rho(t)$ can be designed to achieve a prescribed bound $|e_1(t)|<\phi(t)$. In particular, if $\rho(t)=(\rho_0-\rho_\infty)e^{-\delta t}+\rho_\infty$ and $\lambda>\delta$, then each $|e_i(t)|$ approaches $\{e_i\in\mathbb R:|e_i|\leq2^{i-1}\rho_\infty/\lambda^{k-i}\}$ at least as fast as $e^{-\delta t}$ \cite[Prop.~1]{bechlioulis2016decentralized}.

\begin{remark}
{Unlike traditional funnel designs \cite{4639441,ilchmann2002tracking,berger2025funnel}, the proposed design achieves $\eta_k(t)\to0$ for every funnel-compatible initial condition without requiring $\rho(t)\to0$. The integral term in \eqref{eq:control input} accumulates compensation for $f(x,t)$ and $y_d^{(k)}$. Related smooth error integrals are used in \cite{verginis2024asymptotic,verginis2024barrier}, whereas \cite{verginis2020asymptotic} uses a discontinuous multivariable sign term.}
\end{remark}

\subsection{Stability proof}

{This section provides the funnel-confinement and asymptotic-convergence guarantees of the proposed design.}

\begin{theorem}
Consider the closed-loop system \eqref{eq:system}-\eqref{eq:control input}.
The proposed control design guarantees boundedness of all closed-loop signals and prescribed performance for the filtered error, i.e. $|\eta_k(t)|<\rho(t)$, for all $t\geq 0$. Moreover, there exist positive constants $\bar{e}_{i,1}$, $\bar{e}_{i,2}$, and $\lambda_0\in(0,\lambda)$ such that $|e_i(t)| \le \bar{e}_{i,1}e^{-\lambda_0 t}+\bar{e}_{i,2}$, for all $t\geq 0$, and every $i\in\{1,\dots,k-1\}$.
Furthermore, if the control gains satisfy $k_1>\bar{k}_{1}$ and $k_2>\bar{k}_{2}$, for some positive constants $\bar{k}_{1}$ and $\bar{k}_{2}$ defined in the proof, then $\lim_{t\to\infty}\eta_k(t)=0$ and $\lim_{t\to\infty}e_i(t)=0$, for every $i\in\{1,\dots,k\}$.
\end{theorem}
\begin{proof}
    The proof consists of three steps. We first prove the existence of at least one %
    Filippov solution to the closed-loop system for a maximal time interval $[0,t_{\max})$. Second, we prove the boundedness of $\varepsilon(t)$, leading to $t_{\max} = \infty$, and finally, we prove that $\lim_{t\to\infty}e_i(t) = 0$.
    
   \noindent\textbf{Step 1}. Consider $\sigma \coloneqq \int_0^t \big((k_1 + 1)\alpha  \varepsilon(\tau) + k_2\textup{sgn}(\varepsilon(\tau)) \big)\textup{d}\tau$ and the stacked variable vector $q \coloneqq [\bar{\eta}_{k-1}^\top, \sigma, \xi] \in \mathbb{R}^{k+1}$. {For $k=1$, $\bar\eta_0$ denotes the empty vector and $q=[\sigma,\xi]^\top$.} Differentiation of the elements of $q$ yields
   \begin{subequations}
   \begin{align}
   \dot{\bar\eta}_{k-1}&=-A\bar\eta_{k-1}+B\eta_k,\qquad {k\geq2,}\\
   \dot\sigma&=(k_1+1)\alpha\varepsilon+k_2\operatorname{sgn}(\varepsilon),\\
   \dot\xi&=\frac{f(x,t)+g(x,t)u-y_d^{(k)}-\dot\rho(t)\xi+\Lambda(e)}{\rho(t)}.
   \end{align}
   \end{subequations}
   where $\Lambda(e) \coloneqq \sum_{\ell=1}^{k-1} \binom{k-1}{\ell}\lambda^\ell e_{k-\ell+1}$, and, for $k\geq2$,
  \begin{align} \label{eq:A and B}
  A = \begin{bmatrix}
  \lambda & -1 & 0 & \dots & 0 \\
  0 & \lambda & -1 & \dots & 0 \\
  \vdots & \ddots & \ddots & \vdots & \vdots \\
  0 & 0 & \dots & 0 & \lambda
  \end{bmatrix}, \quad
  B = \begin{bmatrix}0 \\ \vdots \\ 0 \\ 1\end{bmatrix}.
  \end{align}
   {When $k=1$, the first equation and the matrices $A,B$ are absent.}
    {Invoking $x=\bar e+\bar y_d(t)$, $u=-(k_1+1)(\varepsilon-\varepsilon(0))-\sigma$, \eqref{eq:Beta mapping}, and \eqref{eq:ksi and epsilon}, the augmented dynamics can be written as $\dot q=F_q(q,t)$. The plant input $u$ is continuous, but $F_q$ contains $\operatorname{sgn}(\varepsilon)$ through $\dot\sigma$ and is discontinuous at $\varepsilon=0$. Thus, following \cite{filippov1988,cortes2008discontinuous}, the augmented dynamics are}
    \begin{align} \label{eq:q dot}
        \dot{q} \in \mathsf{K}[F_q](q,t),
    \end{align} 
    where $\mathsf{K}[F_q](q,t)$ is the Filippov regularization of $F_q(q,t)$ and replaces $\textup{sgn}(\varepsilon)$ with the set 
    \begin{align} \label{eq:sign set}
        \mathsf{K}[\textup{sgn}](\varepsilon) = \begin{cases}
          \{1\}, & \varepsilon > 0 \\
          \{-1\}, & \varepsilon < 0 \\
          [-1,1], & \varepsilon = 0.
      \end{cases}
  \end{align}
Consider now the open set $\Omega_q \coloneqq \mathbb{R}^{k} \times (-1,1) $. Since $\rho(0) > |\eta_k(0)|$, %
it holds that $q(0) \in \Omega_q$. %
{The regularized right-hand side is measurable and locally essentially bounded on $\Omega_q$. Since $q(0)\in\Omega_q$, standard Filippov existence theory \cite{filippov1988} gives at least one solution on a maximal interval $[0,t_{\max})$. On that interval, $q(t)\in\Omega_q$, $\dot q(t)\in\mathsf K[F_q](q,t)$ almost everywhere, and}
\begin{align}
{|\eta_k(t)|=\rho(t)|\xi(t)|<\rho(t)\leq\rho_{\max}.}
\label{eq:local-funnel}
\end{align}

   \noindent \textbf{Step 2}.  Let now $\zeta \coloneqq (k_1+1)\varepsilon + \sigma$ and note that $\zeta(t)$ is a function of $q$ and is hence well-defined for all $t\in[0,t_{\max})$ due to Step 1. The derivative of $\varepsilon$ is
\begin{align}\label{eq:epsilon dot}
\dot\varepsilon=J_{\mathbb T}(\xi)\dot\xi
=\frac{J_{\mathbb T}(\xi)}{\rho(t)}\big(&f(x,t)+g(x,t)u-y_d^{(k)}
-\dot\rho(t)\xi\nonumber\\[-1mm]
&+\Lambda(e)\big),
\end{align}
where $J_{\mathbb T}(*)\coloneqq\mathrm d\mathbb T(*)/\mathrm d*
=1/(1-*^2)$. Consider the positive definite function
$W(\sigma,\zeta):=\frac12\sigma^2+\frac12\zeta^2$. Since $W$ is
continuously differentiable and the Filippov trajectory is absolutely
continuous, $W(\sigma(t),\zeta(t))$ is absolutely continuous and $\dot W\overset{\mathrm{a.e.}}{\in}\dot{\widetilde W}(\sigma,\zeta)
\coloneqq\zeta\mathsf K[\dot\zeta]+\sigma\mathsf K[\dot\sigma].$ 
Set $\tilde k_1\coloneqq k_1+1$. Algebraic manipulation of the two
terms gives
\begin{align*}
\sigma\mathsf K[\dot\sigma]
&=\sigma\tilde k_1\alpha\varepsilon
  +\sigma k_2\mathsf K[\operatorname{sgn}](\varepsilon)\\
&=\alpha\sigma\zeta-\alpha\sigma^2
  +k_2\sigma\mathsf K[\operatorname{sgn}](\varepsilon)\\
&=\alpha\zeta^2-\alpha\tilde k_1\zeta\varepsilon-\alpha\sigma^2
  +k_2\sigma\mathsf K[\operatorname{sgn}](\varepsilon),\\
\zeta\mathsf K[\dot\zeta]
&=\tilde k_1\zeta\frac{J_{\mathbb T}(\xi)}{\rho(t)}
 \big(f(x,t)+g(x,t)u-y_d^{(k)}-\dot\rho(t)\xi\big)\\[-1mm]
&\quad+\tilde k_1\alpha\zeta\varepsilon
 +\tilde k_1\zeta\frac{J_{\mathbb T}(\xi)}{\rho(t)}\Lambda(e)
 +k_2\zeta\mathsf K[\operatorname{sgn}](\varepsilon).
\end{align*}
Moreover, Lemma~\ref{lemma:1}, $|\eta_k(t)|<\rho_{\max}$ on
$[0,t_{\max})$, and boundedness of $\bar y_d(t)$ imply bounds on
$x(t)$ and $\bar e(t)$ independent of $t_{\max}$. Using the
boundedness of $\dot\rho$ and $\xi$, we can choose a positive constant
$\bar F$, independent of $t_{\max}$, such that
$\bar F\geq|\tilde k_1\rho^{-1}(f-y_d^{(k)}-\dot\rho\xi+\Lambda(e))|$.
Then, for every $w\in\dot{\widetilde W}$ and almost all
$t\in[0,t_{\max})$, it holds that $w\leq-\alpha\sigma^2+\alpha\zeta^2
+k_2(|\sigma|+|\zeta|)
+\tilde k_1\zeta\frac{J_{\mathbb T}(\xi)}{\rho(t)}g(x,t)u
+|\zeta J_{\mathbb T}(\xi)|\bar F.$ Substituting $u=-\zeta+\tilde k_1\varepsilon(0)$ from
\eqref{eq:control input}, using Assumption~\ref{ass:g min}, and
setting $\underline g\coloneqq g_{\min}/\rho_{\max}>0$ and $\bar F_e\coloneqq\bar F+\tilde k_1^2|\varepsilon(0)|
\sup_{t\in[0,t_{\max})}\frac{g(x(t),t)}{\rho(t)},$
yields $w\leq-\alpha\sigma^2+\alpha J_{\mathbb T}(\xi)\zeta^2
+k_2(|\sigma|+|\zeta|)
-\underline g\tilde k_1J_{\mathbb T}(\xi)\zeta^2
+J_{\mathbb T}(\xi)|\zeta|\bar F_e.$ The supremum in $\bar F_e$ is finite because the preceding bounds
confine $x$ to a compact set, $g$ is uniformly bounded there, and $\rho(t)\geq\rho_{\min}>0$. By choosing
$\tilde k_1=k_1+1>\alpha/\underline g$ and defining
$\ell_1\coloneqq\underline g\tilde k_1-\alpha>0$, we obtain $w\leq-\alpha\sigma^2-\ell_1J_{\mathbb T}(\xi)\zeta^2
+k_2(|\sigma|+|\zeta|)
+J_{\mathbb T}(\xi)|\zeta|\bar F_e.$  For positive $q_1,q_2$ satisfying
$q_1<2\alpha$ and $q_1+q_2<2\ell_1$, completing squares gives $k_2(|\sigma|+|\zeta|)
\leq\frac{q_1}{2}(\sigma^2+\zeta^2)+\frac{k_2^2}{q_1}
\leq\frac{q_1}{2}\sigma^2
+J_{\mathbb T}(\xi)\left(\frac{q_1}{2}\zeta^2+\frac{k_2^2}{q_1}\right),$ $
|\zeta|\bar F_e\leq\frac{q_2}{2}\zeta^2
+\frac{\bar F_e^2}{2q_2}.$ Consequently, since $\dot W\in\dot{\widetilde W}$ a.e.,
$\dot W\leq-\tilde\alpha\sigma^2
-J_{\mathbb T}(\xi)
\big(\tilde\ell_1\zeta^2-\bar F_b\big)$ a.e., where $\tilde\alpha\coloneqq\alpha-q_1/2>0$,
$\tilde\ell_1\coloneqq\ell_1-(q_1+q_2)/2>0$, and
$\bar F_b\coloneqq k_2^2/q_1+\bar F_e^2/(2q_2)$.

For completeness, set $r_\zeta\coloneqq\sqrt{\bar F_b/\tilde\ell_1}$.
When $|\zeta|\geq r_\zeta$, the preceding inequality gives
$\dot W\leq-\tilde\alpha\sigma^2\leq0$ a.e. When
$|\zeta|<r_\zeta$, the identity
$\dot\sigma=\alpha(\zeta-\sigma)+k_2v$,
$v\in\mathsf K[\operatorname{sgn}](\varepsilon)$, gives
$D^+|\sigma|\leq-\alpha|\sigma|+\alpha r_\zeta+k_2$. Thus positive
constants $\bar\sigma,\bar\zeta$ exist such that
$|\sigma(t)|\leq\bar\sigma$ and $|\zeta(t)|\leq\bar\zeta$, and hence
$|\varepsilon(t)|\leq\bar\varepsilon$ for some $\bar\varepsilon>0$,
because $\varepsilon=(\zeta-\sigma)/(k_1+1)$. Since
$\xi=\tanh(\varepsilon)$, $|\xi(t)|\leq\bar\xi<1$. Lemma~\ref{lemma:1}
and \eqref{eq:ksi and epsilon} then bound $q$ independently of
$t_{\max}$, i.e., $\|q(t)\|\leq\bar q$ for some $\bar q>0$, and
confine it to a compact subset of $\Omega_q$. The
continuation property therefore excludes finite escape, therefore 
$t_{\max}=\infty$ and all closed-loop signals are bounded. Finally,
\eqref{eq:epsilon dot} shows that
$|\dot\varepsilon(t)|\leq\bar\varepsilon_d$ a.e. for some
$\bar\varepsilon_d>0$.

\noindent\textbf{Step 3}. We now show that $\varepsilon(t)$, and hence
$\xi(t)$ and $\bar e(t)$, converge to zero. Manipulating
\eqref{eq:epsilon dot} gives $m(\bar\eta_{k-1},\varepsilon,t)\dot\varepsilon
=u+\frac{f(x,t)}{g(x,t)}
-\frac{y_d^{(k)}(t)+\dot\rho(t)\xi-\Lambda(e)}{g(x,t)},$ where
\begin{align}\label{eq:m}
m(\bar\eta_{k-1},\varepsilon,t)
&\coloneqq\frac{\rho(t)}{J_{\mathbb T}(\xi)g(x,t)}\nonumber\\[-1mm]
&=\frac{\rho(t)}{J_{\mathbb T}(\tanh(\varepsilon))
g(\bar e+\bar y_d(t),t)}.
\end{align}
Here $\bar e$ is expressed in terms of $\bar\eta_{k-1}$ and
$\varepsilon$ through \eqref{eq:Beta mapping} and
\eqref{eq:ksi and epsilon}. Define $F(\bar\eta_{k-1},\varepsilon,t)
\coloneqq\frac{f(x,t)}{g(x,t)}
-\frac{y_d^{(k)}+\dot\rho\xi-\Lambda(e)}{g(x,t)}.$ Further differentiation yields
\begin{align}\label{eq:m epsilon derivative}
\dot m\dot\varepsilon+m\ddot\varepsilon
=\dot u+T(\bar\eta_{k-1},\varepsilon,t),
\end{align}
where $T(\bar\eta_{k-1},\varepsilon,t)
\coloneqq\frac{\mathrm d}{\mathrm dt}\left(
\frac{f(x,t)}{g(x,t)}-
\frac{y_d^{(k)}+\dot\rho(t)\xi-\Lambda(e)}{g(x,t)}\right).$ To make the filtered error a state defined for every $t$, set $s\coloneqq\frac{u+F}{m}+\alpha\varepsilon$, and 
$\chi\coloneqq[\bar\eta_{k-1}^\top,\varepsilon,s]^\top$. The stated regularity and the bounds from Step~2 make $m,F,u$, and
hence $s$, absolutely continuous along the complete trajectory, while
$m$ is bounded away from zero. Therefore
$s=\dot\varepsilon+\alpha\varepsilon$ a.e. and
$|s(t)|\leq\bar s\coloneqq\bar\varepsilon_d+\alpha\bar\varepsilon$.
Multiplying $\dot s\in\mathsf K[\ddot\varepsilon]
+\alpha\dot\varepsilon$ by $m(\chi,t)$ and using
\eqref{eq:m epsilon derivative} gives $m(\chi,t)\dot s
\in m(\chi,t)\mathsf K[\ddot\varepsilon]
+\alpha m(\chi,t)\dot\varepsilon
=\mathsf K[\dot u]+T(\chi,t)-\dot m(\chi,t)\dot\varepsilon
+\alpha m(\chi,t)\dot\varepsilon.$
Equivalently, $m(\chi,t)\dot s\in-\tfrac12\dot m(\chi,t)s-\varepsilon
+\mathsf K[\dot u]+H(\chi,t)$, where $H(\chi,t)={}T(\chi,t)+\tfrac12\dot m(\chi,t)s
-\dot m(\chi,t)\dot\varepsilon
+m(\chi,t)\alpha\dot\varepsilon+\varepsilon$. Since $\dot u=-\tilde k_1\dot\varepsilon-\tilde k_1\alpha\varepsilon
-k_2\operatorname{sgn}(\varepsilon)
=-\tilde k_1s-k_2\operatorname{sgn}(\varepsilon)$ a.e.,
it becomes
\begin{align}\label{eq:s dot}
m(\chi,t)\dot s\in{}&-\tfrac12\dot m(\chi,t)s-\varepsilon
-\tilde k_1s\nonumber\\[-1mm]
&-k_2\mathsf K[\operatorname{sgn}](\varepsilon)+H(\chi,t).
\end{align}
The decomposition
$H(\chi,t)=H(\chi,t)-H(0,t)+H(0,t)$ gives
$|H(\chi,t)|\leq|H(\chi,t)-H(0,t)|+|H(0,t)|$. Step~2 bounds
$\bar\eta_{k-1},\varepsilon,s$ on a compact tube. On that tube, the
stated regularity makes $H$ uniformly Lipschitz in $\chi$, while
$H(0,\cdot)$ is absolutely continuous and
$H(0,\cdot),\dot H(0,\cdot)\in L_\infty$. Hence, for some $L_H>0$, $|H(\chi,t)-H(0,t)|\leq L_H\|\chi\|,\qquad t\geq0.$ Additionally, if $k_2>\|H(0,\cdot)\|_\infty
+\frac1\alpha\|\dot H(0,\cdot)\|_\infty,$ then \cite[Lemma~1]{xian2004continuous} gives 
$\int_0^tL(\tau)\,\mathrm d\tau\leq\bar h$ for all $t\geq0$, where $L(t)\coloneqq s(t)\big(H(0,t)-k_2\operatorname{sgn}(\varepsilon(t))\big)$ and 
$\bar h\coloneqq k_2|\varepsilon(0)|-\varepsilon(0)H(0,0)$. 

For $k\geq2$, recall from \eqref{eq:A and B} that $A$ is a
$Z$-matrix whose eigenvalues equal $\lambda$, and hence a nonsingular
$M$-matrix. Therefore, \cite[pp.~168]{qu2009cooperative} guarantees a
positive definite diagonal matrix
$C=\operatorname{diag}\{c_{i,i}\}\in\mathbb R^{(k-1)\times(k-1)}$
such that $Q\coloneqq CA+A^\top C>0$.
\footnote{A nonsingular $M$-matrix has nonpositive off-diagonal entries
and positive principal minors.}
Consider
$\chi_P\coloneqq[\chi^\top,\sqrt P]^\top$, with
$P\coloneqq\bar h-\int_0^tL(\tau)\,\mathrm d\tau\geq0$, and define
the set $\Omega_\chi\coloneqq\{\chi_P:|\varepsilon|\leq
\bar\varepsilon,\ |s|\leq\bar s\}$. On $\Omega_\chi$, define $V(\chi_P,t)\coloneqq V_1+V_2,$ where $ V_1\coloneqq\kappa_0\bar\eta_{k-1}^\top C\bar\eta_{k-1}$, $ V_2\coloneqq\tfrac12\varepsilon^2+\tfrac12m(\chi,t)s^2+P,$ and $\kappa_0>0$. For $k=1$, set $V_1\coloneqq0$.

On the compact tube, $0<\underline m\leq m(\chi,t)\leq\bar m$. Therefore positive constants $\underline\lambda,\bar\lambda$ exist
such that $\underline\lambda\|\chi_P\|^2\leq V(\chi_P,t)
\leq\bar\lambda\|\chi_P\|^2.$ For $k\geq2$, differentiation of $V_1$ gives $\dot V_1=-\kappa_0\bar\eta_{k-1}^\top Q\bar\eta_{k-1}
+2\kappa_0\tilde c\,\eta_{k-1}\eta_k,$ where $\tilde c\coloneqq c_{k-1,k-1}$. Since
$|\eta_k|=\rho|\tanh(\varepsilon)|\leq\rho_{\max}|\varepsilon|$,
completing squares yields $\dot V_1\leq-\kappa_0\tilde\kappa_1\|\bar\eta_{k-1}\|^2
+\frac{\kappa_0\tilde c}{\kappa_1}\rho_{\max}^2\varepsilon^2,$ where $0<\kappa_1<\underline Q/\tilde c$,
$\underline Q\coloneqq\lambda_{\min}(Q)$, and
$\tilde\kappa_1\coloneqq\underline Q-\kappa_1\tilde c>0$.

The trajectory and $P$ are absolutely continuous. Thus $V_2$ is
absolutely continuous, its derivative exists a.e., and the chain rule
along the regularized dynamics gives $\dot V_2\overset{\mathrm{a.e.}}{\in}\dot{\widetilde V}_2(\chi_P,t)
\coloneqq\nabla V_2^\top\mathsf K[\dot\chi_P]
+\frac{\partial V_2}{\partial t}.$ Here $\mathsf K[\cdot]$ replaces the signum term in
\eqref{eq:s dot}, and the corresponding term in $\dot P$, by
\eqref{eq:sign set}. Consequently, $\dot{\widetilde V}_2={}-\alpha\varepsilon^2+\varepsilon s
+\tfrac12\dot m(\chi,t)s^2
+m(\chi,t)s\mathsf K[\dot s]+\mathsf K[\dot P],$ and, for a.e. $t\geq0$, $\max_{v\in\dot{\widetilde V}_2}v
\leq{}-\alpha\varepsilon^2-s^2+L_H|s|\|\chi\|-k_1s^2 +s\big(H(0,t)-k_2\mathsf K[\operatorname{sgn}](\varepsilon)\big)
-\mathsf K[L] \leq{}-\alpha\varepsilon^2-s^2+L_H|s|\|\chi\|-k_1s^2.$ For $k\geq2$, combining this inequality with the bound for $\dot V_1$
gives $\dot V\leq L_H|s|\|\chi\|-K\|\chi\|^2-k_1s^2$, where
$K\coloneqq\min\{\kappa_0\tilde\kappa_1,1,
\alpha-\kappa_0\tilde c\rho_{\max}^2/\kappa_1\}>0$ for sufficiently
small $\kappa_0$. For $k=1$, the same inequality holds with
$K\coloneqq\min\{1,\alpha\}$. Finally, $L_H|s|\|\chi\|
\leq\frac{k_1s^2}{2}+\frac{L_H^2\|\chi\|^2}{2k_1},$ and hence $\dot V\leq-K\|\chi\|^2+\frac{L_H^2\|\chi\|^2}{2k_1}
\leq-\left(K-\frac{L_H^2}{2k_1}\right)\|\chi\|^2.$ Choosing $k_1>\frac{L_H^2}{2K}$ therefore yields $\dot V\leq-\gamma\|\chi\|^2$ a.e. for some
$\gamma>0$. For any finite $\tilde r>0$, choose
$\tilde c<\min_{\|\chi_P\|=\tilde r}
\underline\lambda\|\chi_P\|^2$. The corresponding sublevel set
$\Omega_f\coloneqq\{\chi_P\in\mathcal B(0,\tilde r):
\bar\lambda\|\chi_P\|^2\leq\tilde c\}$ is invariant. The constants
$\tilde r$ and $\tilde c$ may be selected large enough to include any
given $P(0)$. Independently of this selection, Step~2 already proves
boundedness of $\varepsilon$ for every initial condition satisfying
$|\eta_k(0)|<\rho(0)$. Integration of the derivative inequality gives
$\chi\in L_2\cap L_\infty$. The
regularized dynamics also give bounded $\dot\chi$, so $\chi$ is
uniformly continuous. Barbalat's lemma yields  $\lim_{t\to\infty}\chi(t)=0$.
Consequently, $\lim_{t\to\infty}\varepsilon(t)=0$, $\lim_{t\to\infty}\xi(t)=\lim_{t\to\infty}\tanh(\varepsilon(t))=0$, and
$\lim_{t\to\infty}\eta_k(t)=\lim_{t\to\infty}\rho(t)\xi(t)=0$. For $k\geq2$, the stable cascade
\eqref{eq:eta_i+1} gives $\lim_{t\to\infty}\bar\eta_k(t)=0$. For $k=1$,
$\eta_1=e_1$. Equation~\eqref{eq:Beta mapping} then gives
$\lim_{t\to\infty}e_i(t)=0$ for every $i\in\{1,\ldots,k\}$. Step~2 has already
established $|\eta_k(t)|<\rho(t)\leq\rho_{\max}$, so
Lemma~\ref{lemma:1} gives the stated error bounds.
\end{proof}

\begin{remark}
The signum term enters the applied input only through the integral state $\sigma$ defined in Step~1. Since $\varepsilon$ and $\sigma$ are absolutely continuous while $|\xi|<1$, $u=-(k_1+1)(\varepsilon-\varepsilon(0))-\sigma$ is absolutely continuous and hence continuous. Nevertheless, $\dot u=-(k_1+1)\dot\varepsilon-(k_1+1)\alpha\varepsilon-k_2\operatorname{sgn}(\varepsilon)$ almost everywhere. Therefore the augmented controller dynamics are discontinuous at $\varepsilon=0$ and interpreted in the Filippov sense \cite{filippov1988}.
\end{remark}

\section{Simulation results}\label{sec:sim}

To assess the proposed control design, we consider the Norrbin model for ship steering \cite{psillakis2016integrator}: %
\begin{align*}
    \dot{x}_1 &= x_2,\\
    \dot{x}_2 &= -\frac{1}{K_1}\left(x_2+a_s x_2^3\right) + \frac{K_2}{K_1}u, \qquad y = x_1,
\end{align*}
{Here $x_1$, $x_2$, and $u$ are the heading, heading rate, and rudder angle, respectively. We use $K_1=21~\mathrm{s}$, $K_2=0.23~\mathrm{s}^{-1}$, $a_s=0.3~\mathrm{s}^2$, and $x_1(0)=x_2(0)=0$. The reference $y_d(t)=0.3491\sin(0.1256t)$ has amplitude $20^\circ$ and period $50$ s. We set $\eta_2=e_2+\lambda e_1$, $\lambda=0.6$, $\rho(t)=(\rho_0-\rho_\infty)e^{-\ell t}+\rho_\infty$, $\rho_0=1$, $\rho_\infty=0.1$, $\ell=0.5$, and $k_1=k_2=20$, $\alpha=5$.}

{Figures~\ref{fig:ship_x1_yd}--\ref{fig:ship_u} show asymptotic tracking, $e_1,e_2\to0$, and confinement of $\eta_2$ within $\pm\rho(t)$. The transient oscillation in $u$ does not constitute classical chattering since the plant input is continuous and the signum term occurs only under the integral. It results from the nonsmooth integral action, the gains used to enforce performance and convergence, and possible numerical-integration amplification.}

\captionsetup[figure]{skip=2pt,belowskip=0pt}
\begin{figure}[tpbh]
    \centering
        \includegraphics[width=\columnwidth]{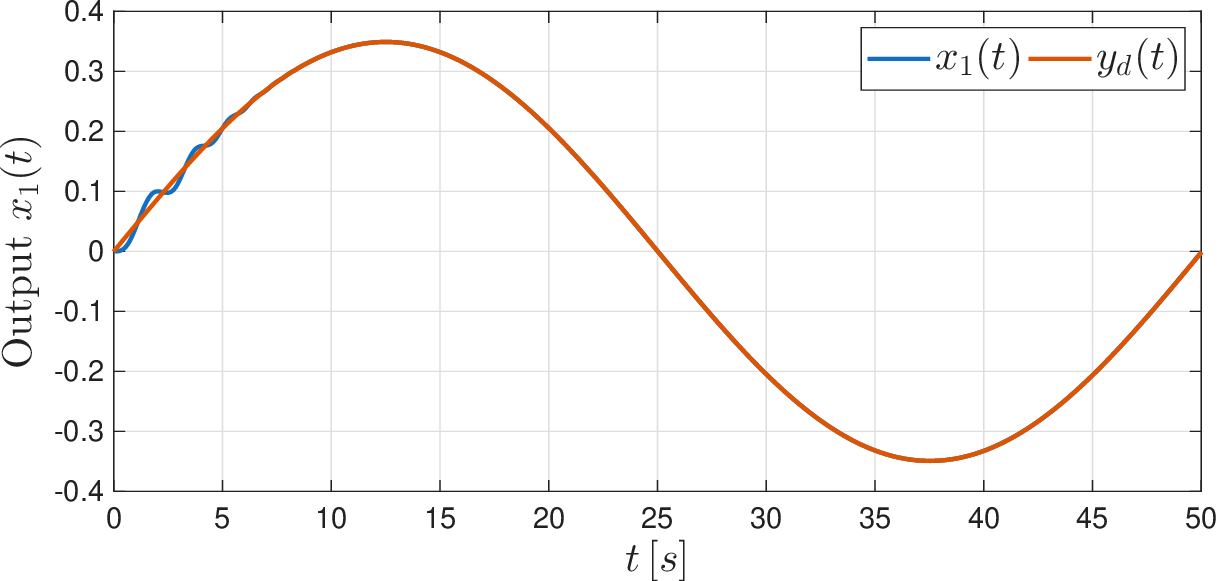}
        \caption{Heading $x_1$ and desired trajectory $y_d$.}
        \label{fig:ship_x1_yd}
        \vspace{-1em}
\end{figure}

\begin{figure}[tpbh]   
    \begin{subfigure}{\columnwidth}
        \centering
        \includegraphics[width=\columnwidth]{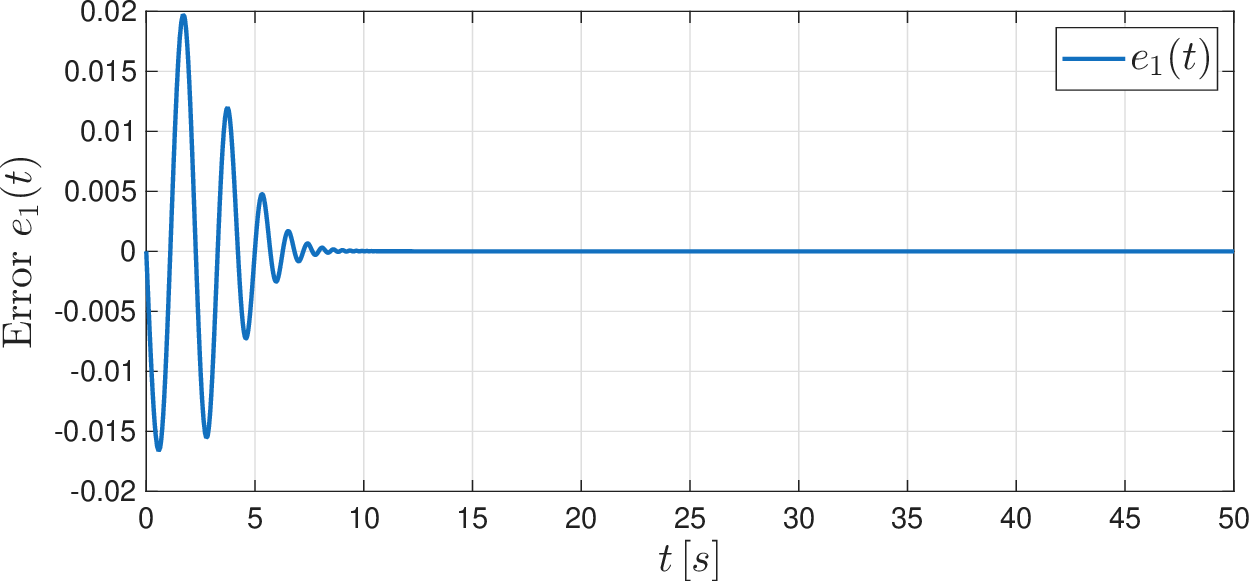}

    \end{subfigure}

    \begin{subfigure}{\columnwidth}
        \centering
        \includegraphics[width=\columnwidth]{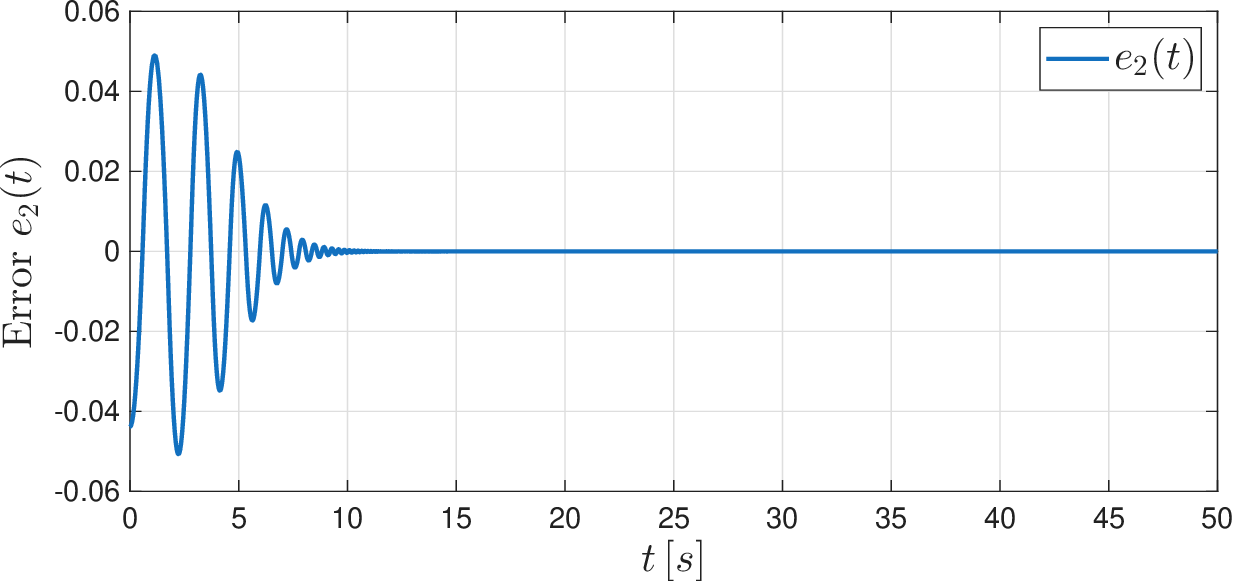}
    \end{subfigure}
        
    \caption{Tracking errors $e_1$ and $e_2$.}\label{fig:ship_e1}
    \vspace{-1em}
\end{figure}

\begin{figure}[tpbh]
    
    \begin{subfigure}{\columnwidth}
        \centering
        \includegraphics[width=\columnwidth]{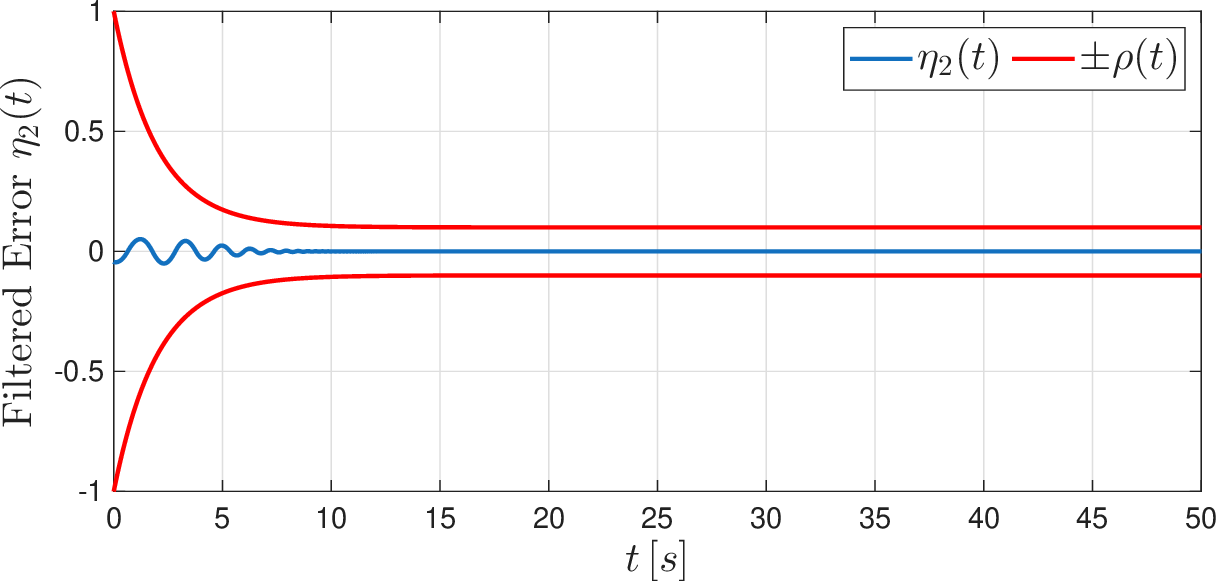}
        
    \end{subfigure}
    \caption{Filtered error $\eta_2$ and prescribed bounds $\pm \rho(t)$.}
\label{fig:ship_eta2_rho}
\vspace{-1em}
    
\end{figure}

\begin{figure}[tpbh]
    
    \begin{subfigure}{\columnwidth}
        \centering
        \includegraphics[width=\columnwidth]{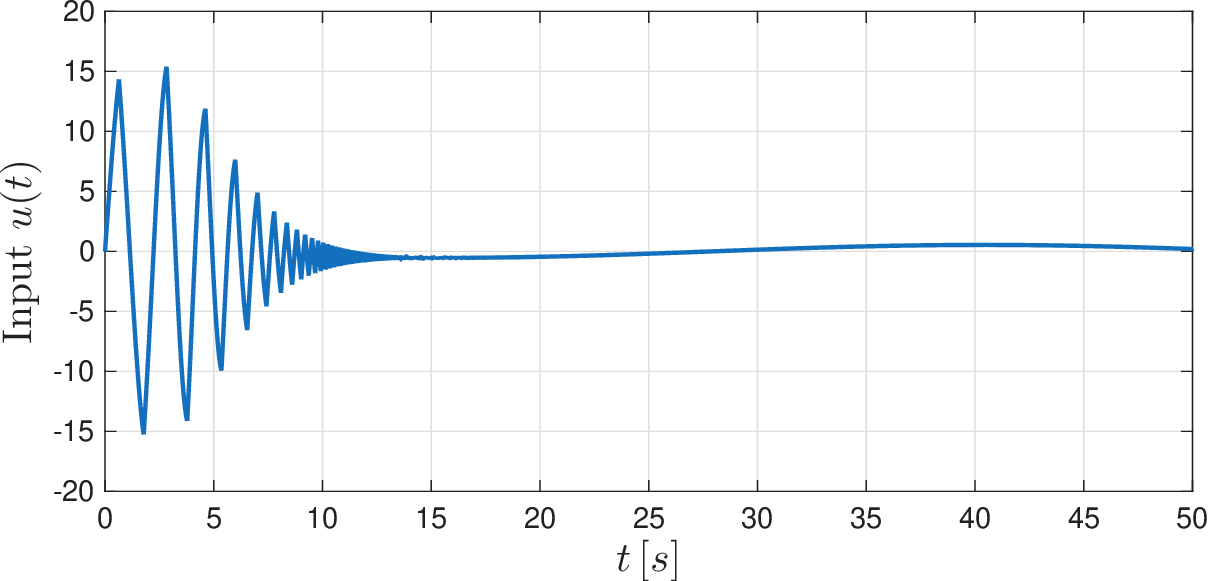}
        
    \end{subfigure}
    
    \caption{Control input $u$.}
        \label{fig:ship_u}
        \vspace{-1em}
\end{figure}

\section{Conclusion and future work}\label{sec:conclusion}

{This paper combined prescribed performance and asymptotic tracking for uncertain high-order nonlinear systems using a continuous, low-complexity controller. Boundedness, funnel invariance, and zero tracking error were established and demonstrated in ship steering. Future work will address multi-input multi-output systems, input constraints, and adaptive gain selection.}

\bibliographystyle{IEEEtran}
\bibliography{references}

@article{bechlioulis2016decentralized,
  title={Decentralized robust synchronization of unknown high order nonlinear multi-agent systems with prescribed transient and steady state performance},
  author={Bechlioulis, Charalampos P and Rovithakis, George A},
  journal={IEEE Transactions on Automatic Control},
  volume={62},
  number={1},
  pages={123--134},
  year={2016},
  publisher={IEEE}
}

@article{verginis2024barrier,
  title={Barrier integral control for global asymptotic stabilization of uncertain nonlinear systems under smooth feedback and transient constraints},
  author={Verginis, Christos K},
  journal={arXiv preprint arXiv:2409.04767},
  year={2024}
}

@article{verginis2024asymptotic,
  title={Asymptotic stabilization of uncertain systems with prescribed transient response via smooth barrier integral control},
  author={Verginis, Christos K},
  journal={IEEE 63rd Conference on Decision and Control (CDC)},
  pages={7284--7289},
  year={2024},
  organization={IEEE}
}

@article{verginis2020asymptotic,
  title={Asymptotic tracking of second-order nonsmooth feedback stabilizable unknown systems with prescribed transient response},
  author={Verginis, Christos K and Dimarogonas, Dimos V},
  journal={IEEE Transactions on Automatic Control},
  volume={66},
  number={7},
  pages={3296--3302},
  year={2021},
  publisher={IEEE}
}

@article{xian2004continuous,
  title={A continuous asymptotic tracking control strategy for uncertain nonlinear systems},
  author={Xian, Bin and Dawson, Darren M and de Queiroz, Marcio S and Chen, Jian},
  journal={IEEE Transactions on Automatic Control},
  volume={49},
  number={7},
  pages={1206--1211},
  year={2004},
  publisher={IEEE}
}

@book{qu2009cooperative,
  title={Cooperative control of dynamical systems: applications to autonomous vehicles},
  author={Qu, Zhihua},
  year={2009},
  publisher={Springer}
}

@article{cortes2008discontinuous,
  title={Discontinuous dynamical systems},
  author={Cortes, Jorge},
  journal={IEEE Control systems magazine},
  volume={28},
  number={3},
  pages={36--73},
  year={2008},
  publisher={IEEE}
}

@article{psillakis2016integrator,
  title={Integrator backstepping with the nonlinear PI method: An integral equation approach},
  author={Psillakis, Haris E},
  journal={European Journal of Control},
  volume={28},
  pages={49--55},
  year={2016},
  publisher={Elsevier}
}

@ARTICLE{4639441,
  author={Bechlioulis, Charalampos P. and Rovithakis, George A.},
  journal={IEEE Transactions on Automatic Control}, 
  title={Robust Adaptive Control of Feedback Linearizable MIMO Nonlinear Systems With Prescribed Performance}, 
  year={2008},
  volume={53},
  number={9},
  pages={2090-2099}}

@article{ilchmann2002tracking,
  title={Tracking with prescribed transient behaviour},
  author={Ilchmann, Achim and Ryan, Eugene P and Sangwin, Christopher J},
  journal={ESAIM: Control, Optimisation and Calculus of Variations},
  volume={7},
  pages={471--493},
  year={2002},
  publisher={EDP Sciences}
}

@ARTICLE{1310482,
  author={Xian, B. and Dawson, D.M. and de Queiroz, M.S. and Chen, J.},
  journal={IEEE Transactions on Automatic Control}, 
  title={A continuous asymptotic tracking control strategy for uncertain nonlinear systems}, 
  year={2004},
  volume={49},
  number={7},
  pages={1206-1211},
  doi={10.1109/TAC.2004.831148}}

@book{filippov1988,
  author={Filippov, Aleksei F.},
  title={Differential Equations with Discontinuous Righthand Sides},
  publisher={Kluwer Academic Publishers},
  address={Dordrecht, The Netherlands},
  year={1988}
}

@article{liu2022asymptotic,
  title={Asymptotic output tracking control with prescribed transient performance of nonlinear systems in the presence of unknown dynamics},
  author={Liu, Yong-Hua and Chen, Lang-Lei and Zhou, Qi and Su, Chun-Yi},
  journal={International Journal of Robust and Nonlinear Control},
  volume={32},
  number={17},
  pages={9363--9379},
  year={2022},
  publisher={Wiley Online Library}
}

@article{deng2022neural,
  title={Neural network-based adaptive asymptotic prescribed performance tracking control of hydraulic manipulators},
  author={Deng, Wenxiang and Zhou, Hang and Zhou, Jin and Yao, Jianyong},
  journal={IEEE transactions on systems, man, and cybernetics: systems},
  volume={53},
  number={1},
  pages={285--295},
  year={2022},
  publisher={IEEE}
}

@article{wang2017rise,
  title={RISE-based asymptotic prescribed performance tracking control of nonlinear servo mechanisms},
  author={Wang, Shubo and Na, Jing and Ren, Xuemei},
  journal={IEEE Transactions on Systems, Man, and Cybernetics: Systems},
  volume={48},
  number={12},
  pages={2359--2370},
  year={2017},
  publisher={IEEE}
}

@article{cao2026achieving,
  title={Achieving Asymptotic Tracking Yet Accommodating More Generalized Performance Specifications via Approximation-free Approach},
  author={Cao, Lan and Huang, Xiucai and Lv, Maolong and Song, Yongduan},
  journal={IEEE Transactions on Automatic Control},
  year={2026},
  publisher={IEEE}
}

@article{fan2017asymptotic,
  title={Asymptotic tracking controller design for nonlinear systems with guaranteed performance},
  author={Fan, Bo and Yang, Qinmin and Jagannathan, Sarangapani and Sun, Youxian},
  journal={IEEE transactions on cybernetics},
  volume={48},
  number={7},
  pages={2001--2011},
  year={2017},
  publisher={IEEE}
}

@article{cao2023neuroadaptive,
  title={Neuroadaptive asymptotic tracking control with guaranteed performance under mismatched uncertainties and saturated inputs},
  author={Cao, Lan and Huang, Xiucai and Ye, Hefu and Song, Yongduan},
  journal={IEEE Transactions on Systems, Man, and Cybernetics: Systems},
  volume={53},
  number={6},
  pages={3784--3794},
  year={2023},
  publisher={IEEE}
}

@article{berger2025funnel,
  title={Funnel control—A survey},
  author={Berger, Thomas and Ilchmann, Achim and Ryan, Eugene P},
  journal={Annual Reviews in Control},
  volume={60},
  pages={101024},
  year={2025},
  publisher={Elsevier}
}

@article{bu2023prescribed,
  title={Prescribed performance control approaches, applications and challenges: A comprehensive survey},
  author={Bu, Xiangwei},
  journal={Asian Journal of Control},
  volume={25},
  number={1},
  pages={241--261},
  year={2023},
  publisher={Wiley Online Library}
}

@article{kanakis2019guaranteeing,
  title={Guaranteeing global asymptotic stability and prescribed transient and steady-state attributes via uniting control},
  author={Kanakis, George S and Rovithakis, George A},
  journal={IEEE Transactions on Automatic Control},
  volume={65},
  number={5},
  pages={1956--1968},
  year={2019},
  publisher={IEEE}
}

@inproceedings{lee2019asymptotic,
  title={Asymptotic tracking via funnel control},
  author={Lee, Jin Gyu and Trenn, Stephan},
  booktitle={2019 IEEE 58th Conference on Decision and Control (CDC)},
  pages={4228--4233},
  year={2019},
  organization={IEEE}
}

\end{document}